\documentclass[twocolumn,aps,pra,longbibliography,superscriptaddress,showkeys,showpacs]{revtex4-2}
\usepackage[bookmarks=false,linkcolor=blue,urlcolor=blue,colorlinks,citecolor=blue]{hyperref}
\usepackage{amsmath, amsthm, amssymb, mathtools, mathrsfs, bbm}
\usepackage{relsize, enumitem, xcolor}

\newcommand{ \myqed }{ \hfill $\blacktriangle$ }

\makeatletter
\def\moverlay{\mathpalette\mov@rlay}
\def\mov@rlay#1#2{\leavevmode\vtop{%
   \baselineskip\z@skip \lineskiplimit-\maxdimen
   \ialign{\hfil$\m@th#1##$\hfil\cr#2\crcr}}}
\newcommand{\charfusion}[3][\mathord]{
    #1{\ifx#1\mathop\vphantom{#2}\fi
        \mathpalette\mov@rlay{#2\cr#3}
      }
    \ifx#1\mathop\expandafter\displaylimits\fi}
\makeatother

\theoremstyle{plain}
\newtheorem{lemma}{Lemma}
\newtheorem{theorem}[lemma]{Theorem}
\newtheorem{remark}[lemma]{Remark}

\begin{document}

\author{Mladen Kova\v{c}evi\'c} 
\email{kmladen@uns.ac.rs}
\affiliation{Faculty of Technical Sciences, University of Novi Sad, Serbia}

\author{Iosif Pinelis} 
\affiliation{Department of Mathematical Sciences, Michigan Technological University, USA}

\author{Marios Kountouris} 
\affiliation{Department of Communication Systems, EURECOM, France}
\affiliation{Department of Computer Science and Artificial Intelligence, University of Granada, Spain}

\title{\Large An information-theoretic analog of the twin paradox}


\begin{abstract}
We revisit the familiar scenario involving two parties in relative motion, in
which Alice stays at rest while Bob goes on a journey at speed $\beta c$ along
an arbitrary trajectory and reunites with Alice after a certain period of time.
It is a well-known consequence of special relativity that the time that passes
until they meet again is different for the two parties and is shorter in Bob's
frame by a factor of $\sqrt{1-\beta^2}$.
We investigate how this asymmetry manifests from an information-theoretic viewpoint.
Assuming that Alice and Bob transmit signals of equal average power to each other
during the whole journey, and that additive white Gaussian noise is present at
both sides, we show that the maximum number of bits per second that Alice can
transmit reliably to Bob is always higher than the one Bob can transmit to Alice.
Equivalently, the energy per bit invested by Alice is lower than that invested
by Bob, meaning that the traveler is less efficient from the communication perspective,
as conjectured by Jarett and Cover.
\end{abstract}

\keywords{Doppler effect, special relativity, information theory, channel capacity, Gaussian noise}


\maketitle

\section{Introduction}

The process of communication involves representing the message to be transmitted
by a physical quantity that the intended recipient can measure and thereby infer
the message.
As Shannon showed in his pioneering work \cite{shannon}, in most cases of interest,
messages can be encoded in the chosen physical quantity in such a way that the
receiver can recover them with accuracy as high as desired, even in the presence
of noise and other signal distortions, while the information rate -- the average
number of transmitted bits per second -- remains bounded away from zero.
For a given level of noise, however, there is a natural upper bound on the number
of bits per second that can be reliably transmitted to the receiver, and this upper
bound depends on the physical parameters such as signal power, signal duration, and
frequency bandwidth.
Since the transmitter and the receiver are located at different points in space-time,
and potentially in different reference frames, the fact that all of these quantities
are relative, i.e., observer-dependent, must in general be taken into account when
attempting to determine the fundamental limits of communication systems, or to develop
practical systems that operate as desired.

In this work, we study the limits of (classical) information transfer within the
framework of special theory of relativity~\cite{einstein, rindler}, more specifically
within the model introduced by Jarett and Cover~\cite{jarett+cover}.
The mentioned paper investigates how the inherent asymmetry between the two observers
in the famous twin problem~\cite{einstein, rindler, darwin}, which implies that the
traveler will be younger than his stay-at-home twin upon his return from the journey,
manifests in an information-theoretic setting where it is assumed that the twins are
communicating with each other during the whole journey.
Among the very few works continuing this line of research, some have studied the
same model from the aspect of security in the presence of an eavesdropper~\cite{garcia},
while others extended the model to multiuser settings~\cite{garcia2, ren+xu}.
We continue the work initiated in~\cite{jarett+cover} along a different direction and
focus on an intriguing conjecture made therein, stating that, under any symmetric
transmission constraints, the traveling twin is always less efficient from the
communication perspective.
This statement, which can viewed as an information-theoretic version of the twin
paradox, was proved in~\cite{jarett+cover} in several special cases.
Its most general form, however, appears to be a highly nontrivial problem.

\textit{Results:}
Our main result is a proof of the conjecture just mentioned, under the assumption
that the traveler is moving at constant, but otherwise arbitrary speed, along an
arbitrary closed trajectory.
On the technical side, our contribution consists of recasting the problem into a
probabilistic statement and applying the results of extremal probability theory.
Along the way we prove an inequality that is more general than what the conjecture
requires, and that may be of separate interest in information theory
(see Theorem~\ref{thm:inequality} and Remark~\ref{rem:KL}).%

\textit{Relevance and context:}
Relativistic effects on information transmission and processing tasks have been
analyzed in various contexts, both in quantum and classical domains; see, e.g.,
\cite{checinska+dragan, jonsson, kent, kent2, peres+terno, peres+terno2}.
Studying this interplay is of great theoretical value \cite{peres+terno2}, while
also having direct applications in GPS \cite{ashby2} and other distributed
systems~\cite{messerschmitt}.
We view this work as part of this general quest to understand the interplay
between information theory and the theory of relativity, and, more specifically,
to understand the fundamental limits of communication in relativistic settings.
The concrete scenario studied in the paper is convenient for this purpose as the
twin paradox is one of the most illustrative and counter-intuitive consequences
of special relativity.

\section{Model description and problem formulation}

In this section we describe the model introduced in~\cite{jarett+cover},
which is adopted in this paper as well, and state some of the results obtained
in~\cite{jarett+cover} that are relevant for the specific problem we intend to
study.

\subsection{The twin problem and the Doppler effect}

We consider the standard ``twin'' scenario \cite[Sec.\ 3.5]{rindler}: two observers
in Minkowski space, of which one, called Alice, is assumed to be stationary, while
the other, called Bob, goes on a journey and reunites with Alice after a certain
period of time.
(The assumption that Bob's initial and final point coincide with Alice is in fact
not necessary; everything said in the sequel holds for any closed trajectory.)
Bob's trajectory is an arbitrary piecewise smooth closed curve, with finitely many
singular points, i.e., cusps, at which the velocity vector is discontinuous.
We emphasize that singular points are allowed only for the purpose of generality
and do not affect the analysis in any way.
Suppose that the ideal clocks belonging to Alice and Bob, both set to $0$ at the
start of the journey, are reading $ T_A $ and $ T_B $, respectively, when the journey
is completed, and let
\begin{equation}
\label{eq:gamma}
  \gamma \coloneqq \frac{T_A}{T_B} .
\end{equation}
(`$ \coloneqq $' means that the symbol on the left-hand side is \emph{defined}
by the expression on the right-hand side.)
The so-called twin paradox is the statement that $\gamma > 1$.

Let $ \alpha_A(t) $ denote the Doppler factor observed by Bob for a pulse transmitted
by Alice at time $ t $ in her frame (i.e., a photon of frequency $ \nu $ emitted by
Alice at time $ t $ is received by Bob as a photon of frequency $ \nu \alpha_A(t) $).\linebreak
By the definition of the Doppler factor, we know that two pulses transmitted at time
instants $ t $ and $ t + dt $ in Alice's frame will be received
\begin{equation}
\label{eq:dtdtau}
  d\tau = \frac{1}{\alpha_A(t)} dt
\end{equation}
seconds apart in Bob's frame.
Likewise, denote by $ \alpha_B(\tau) $ the Doppler factor observed by Alice
for a pulse transmitted by Bob at Bob's time $ \tau $.
Let $ \beta(\tau) c $ be Bob's speed as measured in Alice's frame at Bob's
time $ \tau $, and $ \beta_r(\tau) $ the radial component of $ \beta(\tau) $
in Alice's frame.
Then, since no gravitational fields exist in the vicinity of either of the observers,
the Doppler factors can be expressed as \cite[Sec.\ 4.3]{rindler},\cite{jarett+cover}
\begin{subequations}
\label{eq:doppler}
\begin{align}
  \alpha_A(t)    &= \frac{1 - \beta_r(\tau_B(t))}{\sqrt{1 - \beta^2(\tau_B(t))}} ,  \\
	\alpha_B(\tau) &= \frac{\sqrt{1 - \beta^2(\tau)}}{1 + \beta_r(\tau)} .
\end{align}
\end{subequations}
Here $ \tau_B(t) $ denotes the moment of reception in Bob's frame of a pulse sent
at moment $ t $ in Alice's frame, which can be written in the form (see \eqref{eq:dtdtau})
\begin{equation}
\label{eq:tauprim}
  \tau_B(t) = \int_{0}^t \frac{1}{\alpha_A(u)} \,du .
\end{equation}
The following useful relations follow immediately \cite{jarett+cover}
\begin{subequations}
\label{eq:intalpha}
\begin{align}
\label{eq:intalphaA}
  \tau_B(T_A) =  \int_{0}^{T_A} \frac{1}{\alpha_A(t)} \,dt   &= T_B ,  \\
\label{eq:intalphaB}
  \int_{0}^{T_B} \frac{1}{\alpha_B(\tau)} \,d\tau   &= T_A .
\end{align}
\end{subequations}

\subsection{Communication model -- relativistic\\AWGN channel}

Suppose that Alice is sending information to Bob by using signals of instantaneous
power $ P_A(t) $, duration $ T_A $, and (approximate) frequency bandwidth $ W_A $,
all specified in her (i.e., the transmitter's) reference frame, and that additive
white Gaussian noise of power spectral density $ \eta_B $ is distorting the signal
at the receiving side.
Then the maximum transmission rate, i.e., the maximum number of bits per second
in the transmitter's frame, at which the receiver is still able to recover the
information with an arbitrarily small error probability, is given by
\begin{equation}
\label{eq:capacity}
  \overline{C}_A
	 = \frac{1}{T_A} \int_{0}^{T_A} W_A \log\!\left(1 + \alpha_A(t)\frac{P_A(t)}{\eta_B W_A} \right) dt ,
\end{equation}
where $ \log $ is the base-$ 2 $ logarithm.
This was shown in \cite{jarett+cover} and follows directly from Shannon's
theory \cite{shannon2} after taking into account appropriate relativistic
modifications of the quantities $ P_A(t), T_A, W_A $ at the receiving side.
(Strictly speaking, \eqref{eq:capacity} is the expression for the maximum
transmission rate in a channel in which the function $ \alpha_A(t) $ is
periodic of period $ T_A $, which is the case if, e.g., Alice is stationary
while Bob is moving in a periodic fashion.
Namely, the notion of channel capacity is asymptotic and in most cases it
is necessary that signal duration tends to infinity in order to achieve
arbitrarily small error probability.)

The function $ P_A(t) $ can be chosen by Alice subject to the imposed
transmission constraints, the most natural and frequently assumed constraint
being of the form
\begin{equation}
\label{eq:constraint}
  \frac{1}{T_A} \int_{0}^{T_A} P_A(t) \,dt = \overline{P}_A ,
\end{equation}
for a given $ \overline{P}_A > 0 $.
The optimal choice that maximizes the expression in \eqref{eq:capacity} is
then given by
\begin{equation}
\label{eq:power}
  P^*_A(t) = \max\!\left\{ 0, \; \lambda_A - \frac{\eta_B W_A}{\alpha_A(t)} \right\} ,
\end{equation}
where $ \lambda_A $ is chosen so that the constraint \eqref{eq:constraint}
is satisfied.
To simplify the analysis, we shall assume that the available power $ \overline{P}_A $
is sufficiently large so that transmission is possible during the whole trip,
i.e., that
\begin{equation}
\label{eq:power2}
  \lambda_A \geqslant \max_{t} \frac{\eta_B W_A}{\alpha_A(t)} .
\end{equation}

Similarly, if Bob is sending information to Alice over an additive white Gaussian
noise channel with parameters $ P_B(t), T_B, W_B, \eta_A $, we can write the same
expressions \eqref{eq:capacity}--\eqref{eq:power2} with subscripts $ A $ and $ B $
interchanged. 

We emphasize that it is assumed in the model that Bob's trajectory and velocity
are known in advance and, consequently, both parties:
1) can predict each other's position and aim correspondingly so that, ideally,
the entire signal is intercepted and no power is lost,
2) know the Doppler factors affecting the transmission at every instant and can
adjust the power appropriately so as to achieve optimal transmission strategy.

Suppose then that Alice and Bob transmit signals to each other using optimal
power adaptation, which is by \eqref{eq:power} given by
\begin{subequations}
\label{eq:powerAB}
\begin{align}
\label{eq:powerA}
  P^*_A(t) &= \overline{P}_A + \eta_B W_A \left( \frac{1}{\gamma} - \frac{1}{\alpha_A(t)} \right) ,  \\
\label{eq:powerB}
  P^*_B(t) &= \overline{P}_B + \eta_A W_B \left( \gamma - \frac{1}{\alpha_B(t)} \right) ,
\end{align}
\end{subequations}
assuming that (see \eqref{eq:power2})
\begin{subequations}
\label{eq:pnw}
\begin{align}
  \frac{\overline{P}_A}{\eta_B W_A}  &\geqslant  \max_{t} \frac{1}{\alpha_A(t)} - \frac{1}{\gamma} ,  \\
	\frac{\overline{P}_B}{\eta_A W_B}  &\geqslant  \max_{\tau} \frac{1}{\alpha_B(\tau)} - \gamma .
\end{align}
\end{subequations}
The expressions for their transmission rates then become%
\begin{subequations}
\label{eq:capacityAB}
\begin{align}
\label{eq:capacityA}
  \overline{C}_A
	  &=  \frac{W_A}{T_A} \int_{0}^{T_A} \log\!\left(\alpha_A(t) \left( \frac{\overline{P}_A}{\eta_B W_A} + \frac{1}{\gamma} \right) \right) dt ,  \\
\label{eq:capacityB}
  \overline{C}_B
	  &=  \frac{W_B}{T_B} \int_{0}^{T_B} \log\!\left(\alpha_B(\tau) \left( \frac{\overline{P}_B}{\eta_A W_B} + \gamma \right) \right) d\tau .
\end{align}
\end{subequations}

\subsection{The Jarett--Cover conjecture}

A conjecture was made in \cite[Sec.~VIII]{jarett+cover} that, everything else
being symmetric, namely
\begin{align}
\label{eq:symmetric}
  \overline{P}_A = \overline{P}_B , \quad   W_A = W_B , \quad  \eta_A = \eta_B ,
\end{align}
it always holds that
\begin{align}
\label{eq:conjecture}
  \overline{C}_A > \overline{C}_B ,
\end{align}
i.e., the maximum number of bits per second that Alice can transmit reliably
to Bob is higher than the one Bob can transmit to Alice.

The inequality \eqref{eq:conjecture} was shown in \cite[Sec.~IV and Sec.~IX-D]{jarett+cover}
for the special cases of purely circular (Bob moving on a circular orbit around
Alice) and purely radial (Bob moving away from Alice along a straight line, and
coming back the same way) constant-speed motion.
We shall demonstrate in the following section that the statement is true for an
arbitrary trajectory.
Proving this will require a different strategy and a more elaborate analysis
compared to the cases of circular and radial motion in which \eqref{eq:conjecture}
can be directly established by elementary means.

The conjecture stated above can also be phrased in terms of average energy
invested per one transmitted bit, which is a natural measure of efficiency
used in communication theory.
Denoting by $ E = \overline{P} T $ the energy of the signal and by
$ N = \overline{C} T $ the total number of transmitted bits, the energy
per bit is simply $ E / N = \overline{P} / \overline{C} $.
Assuming that $ \overline{P}_A = \overline{P}_B $, we have
\begin{equation}
  \frac{E_A/N_A}{E_B/N_B}
	 = \frac{(\overline{P}_A T_A)/(\overline{C}_A T_A)}{(\overline{P}_B T_B)/(\overline{C}_B T_B)}
	 = \frac{\overline{C}_B}{\overline{C}_A} ,
\end{equation}
so \eqref{eq:conjecture} is equivalent to saying that Alice is more efficient than Bob in the sense that
$ E_A / N_A < E_B / N_B $.

In the rest of the article, we assume that \eqref{eq:symmetric} holds
and we write these quantities without subscripts.

\section{Proof of the Jarett--Cover conjecture: Constant-speed case}

We assume throughout this section that Bob is moving at speed
$ \beta(\tau) = \beta = \textnormal{const} $ with respect to Alice, $\beta \in (0,1)$.
In this case the Doppler factors \eqref{eq:doppler} take the form
\begin{subequations}
\label{eq:doppler2}
\begin{align}
  \alpha_A(t)    &= \gamma \big(1 - \beta_r(\tau_B(t)) \big) ,   \\
  \alpha_B(\tau) &= \frac{1}{\gamma \big(1 + \beta_r(\tau) \big)} ,
\end{align}
\end{subequations}
where
\begin{equation}
  \gamma = \frac{T_A}{T_B} = \frac{1}{\sqrt{1-\beta^2}}
\end{equation}
is the aging factor \eqref{eq:gamma}.

For future reference, we note that the integral of the radial component of
the velocity, which is equal to the total displacement of Bob with respect
to Alice during the trip, must be zero since the trajectory is a closed curve,
i.e.,
\begin{equation}
\label{eq:betar}
  \int_{0}^{T_B} \beta_r(\tau) \,d\tau = 0 .
\end{equation}
Let us also denote
\begin{equation}
\label{eq:b}
  b \coloneqq \max_{\tau} |\beta_r(\tau)| .
\end{equation}
Note that $ b \leqslant \beta $, where inequality may be strict in general.
For example, if Bob's trajectory is circular with Alice in the center, the
radial component of the velocity is equal to zero at every instant, and so
$ b = 0 $.

Substituting \eqref{eq:doppler2} into \eqref{eq:capacityAB} and introducing
the change of variables $ t \mapsto \tau = \tau_B(t) $ in \eqref{eq:capacityA},
in which case $ dt = \alpha_A(t) d\tau = \gamma\big(1 - \beta_r(\tau)\big) d\tau $
(see \eqref{eq:dtdtau}), equations \eqref{eq:capacityAB} become
\begin{subequations}
\label{eq:capacity2}
\begin{align}
\nonumber
  \overline{C}_A  &=  \frac{W}{T_A} \int_{0}^{T_B} \gamma\big(1 - \beta_r(\tau)\big) \cdot  \\
\label{eq:capacity2A}
	&\phantom{xxxxxxx}\cdot\log\!\left(\big(1 - \beta_r(\tau)\big) \left( \gamma \frac{\overline{P}}{\eta W} + 1 \right) \right) d\tau ,  \\
\label{eq:capacity2B}
  \overline{C}_B  &=  \frac{W}{T_B} \int_{0}^{T_B} \log\!\left(\big(1 + \beta_r(\tau)\big)^{-1} \left( \frac{1}{\gamma}\frac{\overline{P}}{\eta W} + 1 \right) \right) d\tau ,
\end{align}
\end{subequations}
where we shall assume that (see \eqref{eq:pnw}, \eqref{eq:doppler2} and \eqref{eq:b})
\begin{align}
\label{eq:pnw2}
  \frac{\overline{P}}{\eta W}
	 \geqslant  \max\!\left\{ \frac{1}{\gamma} \frac{b}{1 - b} ,\; \gamma b \right\} .
\end{align}
To further simplify the notation, define
\begin{equation}
\label{eq:sigmaf}
  \sigma \coloneqq \frac{\overline{P}}{\eta W}
	 \quad \textnormal{and} \quad
  f(x) \coloneqq \beta_r(T_B x) , \; x \in [0,1]
\end{equation}
and note that, by \eqref{eq:betar} and \eqref{eq:b}, the function $ f $ satisfies
the conditions $ |f(x)| \leqslant b $ and $ \int_{0}^{1} f(x) dx = 0 $.
Then \eqref{eq:capacity2} becomes
\begin{subequations}
\label{eq:capacity3}
\begin{align}
\label{eq:capacity3A}
  \overline{C}_A  &=  W \int_{0}^{1} (1 - f(x)) \log\!\big((1 - f(x)) ( \gamma \sigma + 1 ) \big) \,dx ,  \\
\label{eq:capacity3B}
  \overline{C}_B  &=  W \int_{0}^{1} \log\!\big((1 + f(x))^{-1} ( \gamma^{-1} \sigma + 1 ) \big) \,dx ,
\end{align}
\end{subequations}
so the inequality $ \overline{C}_A > \overline{C}_B $ is equivalent to
\begin{multline}
\label{eq:capacity4}
   \int_{0}^{1} (1 - f(x)) \log(1 - f(x)) \,dx  +  \int_{0}^{1} \log(1 + f(x)) \,dx  \\
	 \,>\, \log\frac{\gamma^{-1} \sigma + 1}{\gamma \sigma + 1}
\end{multline}
for $ \sigma \geqslant \max\!\big\{ \frac{1}{\gamma} \frac{b}{1 - b} , \, \gamma b \big\} $.
In Theorem \ref{thm:inequality} below we state an inequality stronger than
\eqref{eq:capacity4}.
Proving the theorem will therefore complete the proof of our claim that
$ \overline{C}_A > \overline{C}_B $.

\begin{theorem}
\label{thm:inequality}
Let $ f \colon [0,1] \to \mathbb{R} $ be a function satisfying
$ |f(x)| \leqslant b < 1 $ and $ \int_{0}^{1} f(x) dx = 0 $.
Then
\begin{multline}
\label{eq:capacity5}
   \int_{0}^{1} \big( (1 - f(x)) \log(1 - f(x)) + \log(1 + f(x)) \big) \,dx  \\
	 \,>\, \log\!\left( 1 - b^3 \right) .
\end{multline}
\end{theorem}

Let us first show that the inequality \eqref{eq:capacity5} is indeed stronger
than \eqref{eq:capacity4}.

\begin{lemma}
\label{thm:p}
For every $ \beta \in (0, 1) $, $ b \in [0, \beta] $, and
$ \sigma  \geqslant  \max\!\big\{ \frac{1}{\gamma} \frac{b}{1 - b} , \, \gamma b \big\} $,
where $ \gamma = (1-\beta^2)^{-1/2} $, we have
\begin{equation}
\label{eq:p}
  \log(1 - b^3)  \geqslant  \log\frac{\gamma^{-1} \sigma + 1}{\gamma \sigma + 1} .
\end{equation}
\end{lemma}
\begin{proof}[Proof of Lemma \ref{thm:p}]
Since $ \gamma > 1 $, the expression $ \frac{\gamma^{-1} \sigma + 1}{\gamma \sigma + 1} $
is a monotonically decreasing function of $ \sigma $, so it is enough to show
that $ 1 - b^3  \geqslant \frac{\gamma^{-1} \sigma + 1}{\gamma \sigma + 1} $
for the smallest admissible value of $ \sigma $, which is
$ \max\!\big\{ \frac{1}{\gamma} \frac{b}{1 - b} , \, \gamma b \big\} $.
Consider two cases:
\begin{enumerate}[leftmargin=*]
\item
If $ \gamma b > \frac{1}{\gamma}\frac{b}{1 - b} $, i.e., $ b < \beta^2 $, then
\begin{multline}
\label{eq:case1}
  \frac{\gamma^{-1} \sigma + 1}{\gamma \sigma + 1}
	  \leqslant  \frac{\gamma^{-1} \gamma b + 1}{\gamma \gamma b + 1}
		=          \frac{b + 1}{\frac{1}{1 - \beta^2} b + 1}
		<          \frac{b + 1}{\frac{1}{1 - b} b + 1}  \\
		=          1 - b^2
		<          1 - b^3 .
\end{multline}
\item
If $ \gamma b \leqslant \frac{1}{\gamma}\frac{b}{1 - b} $, then
\begin{multline}
\label{eq:case2}
  \frac{\gamma^{-1} \sigma + 1}{\gamma \sigma + 1}
	  \leqslant  \frac{\frac{1}{\gamma} \frac{1}{\gamma}\frac{b}{1 - b} + 1}{\gamma \frac{1}{\gamma}\frac{b}{1 - b} + 1}
		=          \frac{(1 - \beta^2)\frac{b}{1 - b} + 1}{\frac{1}{1 - b}}  \\
		=          1 - b \beta^2
		\leqslant  1 - b^3 .
\end{multline}
\end{enumerate}
The inequality \eqref{eq:p} is thus proved.
It is also evident from \eqref{eq:case2} that equality holds in \eqref{eq:p}
only in the boundary case when $ b = \beta $ and
$ \sigma = \frac{1}{\gamma} \frac{b}{1 - b} = \beta\big(\frac{1+\beta}{1-\beta}\big)^{1/2} $.
\end{proof}

\begin{remark}
\label{rem:KL}
\textnormal{
The inequality stated in Theorem \ref{thm:inequality} is stronger than
what is needed to conclude that $ \overline{C}_A > \overline{C}_B $ not
only because the constant on the right-hand side is
$ \geqslant\!\log\frac{\gamma^{-1} \sigma + 1}{\gamma \sigma + 1} $, but
because the inequality holds for \emph{any} function $ f $ satisfying
$ |f(x)| \leqslant b < 1 $ and $ \int_{0}^{1} f(x) dx = 0 $, regardless
of whether or not this function may represent the radial component
of the velocity along some closed trajectory.}

\textnormal{
Note also that, under the stated conditions, $ 1 - f $ and $ 1 + f $ are
probability densities on $ [0, 1] $, so the inequality \eqref{eq:capacity5}
can also be written in the form
\begin{align}
\label{eq:KL}
  D_{\textsc{kl}}(1-f \parallel 1) - D_{\textsc{kl}}(1 \parallel 1+f)  \,>\,  \log(1 - b^3) ,
\end{align}
which may be of separate interest in information theory.
Here $ D_{\textsc{kl}}(\cdot\!\parallel\!\cdot) $ denotes the Kullback--Leibler
divergence (relative entropy).
\myqed}
\end{remark}

\begin{proof}[Proof of Theorem \ref{thm:inequality}]
We first restate the problem in probabilistic language.
Namely, the inequality \eqref{eq:capacity5} can be rewritten (or, in fact,
generalized) as 
\begin{equation}
\label{eq:E}
	\mathbb{E}[g(Y)]  >  \ln(1 - b^3) ,
\end{equation}
where 
\begin{equation}
\label{eq:g}
  g(x) \coloneqq \ln(1+x) + (1-x)\ln(1-x)
\end{equation}
and $ Y $ is a random variable such that 
\begin{equation}
\label{eq:PY}
	\mathbb{P}(|Y| \leqslant b) = 1  \quad\text{and} \quad  \mathbb{E}[Y] = 0 .
\end{equation}
To see this, take $ Y \coloneqq f(X) $, where $ X $ is a random variable
uniformly distributed on $ [0, 1] $.

By well-known results in probability theory \cite{hoeffding} (see also
\cite{karr} or \cite[Cor.\ 13]{pinelis}), it is sufficient to prove the
inequality \eqref{eq:E} for random variables $ Y $ whose distributions
are supported on sets of cardinality $ \leqslant\!2 $.
In other words, without loss of generality, we may assume that the distribution
of $ Y $ is supported on a set $ \{ u, v\} \subset [-b, b] $.
Given such $ u $ and $ v $, if $ b $ is now replaced by the smallest possible
value of $ b_0 $ such that $ \{u, v\} \subseteq [-b_0, b_0] $, the left-hand
side of \eqref{eq:E} will not change whereas the right-hand side will not
decrease.
This means that, without loss of generality, we may further assume that
one of the values $ u $ or $ v $ coincides with one of the endpoints of
the interval $ [-b, b] $; that is, either $ v = -b $ or $ v = b $.
To cover both of these cases at once, let us assume that the distribution
of $ Y $ is supported on a set of the form $ \{ -b, u, b \} $ with $ u \in (-b, b) $,
and denote $ p \coloneqq \mathbb{P}(Y = -b) $, $ q \coloneqq \mathbb{P}(Y = u) $,
$ r \coloneqq \mathbb{P}(Y = b) $.

By what was said in the previous paragraph, it remains to show that  
\begin{equation}
\label{eq:L}
	L \coloneqq pg(-b) + qg(u) + rg(b) - \ln(1-b^3)  >  0
\end{equation}
given the conditions
\begin{subequations}
\label{eq:pqr}
\begin{align}
	-b <\ &u < b ,  \\
\label{eq:pqr2}
	p, q, r \geqslant 0 , \quad  &p + q + r = 1 ,  \\
\label{eq:pqr3}
	-pb + qu &+ rb = 0 .
\end{align}
\end{subequations}
Solving \eqref{eq:pqr2} and \eqref{eq:pqr3} for $ p $ and $ r $ gives
$ p = \frac{1}{2}\big(1 - q\big(1 - \frac{u}{b}\big)\big) $,
$ r = \frac{1}{2}\big(1 - q\big(1 + \frac{u}{b}\big)\big) $,
and after substituting these expressions into \eqref{eq:L} we get 
\begin{equation}
	L = M(q,b,u) \coloneqq R(b) + qS(b,u) ,
\end{equation}
where 
\begin{align}
	R(b)   &\coloneqq  \ln\!\left(1-b^2\right) - \ln\!\left(1-b^3\right) + \frac{b}{2} \ln\frac{1+b}{1-b}  \\ 
	S(b,u) &\coloneqq  g(u) - (1-u/2) \log\!\left(1-b^2\right) - \frac{b}{2} \ln\frac{1+b}{1-b} .
\end{align}
Since $ M(q,b,u) $ is affine in the variable $ q $, in order to prove
\eqref{eq:L} under the conditions \eqref{eq:pqr} it suffices to show that 
\begin{equation}
\label{eq:M}
	M(0,b,u) > 0  \quad\text{and}\quad  M(1,b,u) > 0
\end{equation}
for $ u \in (-b, b) $ and $ b \in (0, 1) $. 

We have $ M(0,b,u) = R(b) $.
The derivatives of this function are given by
\begin{align}
\label{eq:50}
  R'(b)  &=  -\frac{b}{1-b^2} + \frac{3b^2}{1-b^3} + \frac{1}{2} \ln\frac{1+b}{1-b} ,  \\
\label{eq:60}
	R''(b) &=  \frac{b \left(6+10 b+2 b^2-3 b^3+2 b^4+b^5\right)}{(1-b)^2 (1+b)^2 \left(1+b+b^2\right)^2} .
\end{align}
Obviously, $ R''(b) > 0 $ for $ b \in (0, 1) $, so $ R(b) $ is convex.
Additionally, $ R(0) = R'(0) = 0 $, so it must be the case that $ M(0,b,u) = R(b) > 0 $
for any $ b \in (0, 1) $. 

Next, 
\begin{multline}
	M(1,b,u)
	 =          F(b)
	 \coloneqq  F(b,u)  \\
	 \coloneqq  g(u) - \ln\!\left(1 - b^3\right) + \frac{u}{2} \ln\!\left(1 - b^2\right) , 
\end{multline}
Differentiating with respect to $ b $ we find that
\begin{equation}
\begin{aligned}
	F'(b) &\cdot b^{-1} (1 - b) (1 + b) (1 + b + b^2)  \\
	 &= 3 b (1 + b) - (1 + b + b^2) u  \\ 
	 &\geqslant 3 b (1 + b) - (1 + b + b^2) b  \\
	 &= b (2 + 2 b - b^2)  \\
	 &> 0 . 
\end{aligned}
\end{equation}
So, $ F $ is increasing in $ b $ and hence $ M(1,b,u) = F(b,u) > F(|u|,u) $,
because $ |u| < b $.
If now $ u \in [0, 1) $, then $ F(|u|,u) = F(u,u) = R(u) \geqslant 0 $, by
what was shown in the previous paragraph. 
Finally, if $ u \in (-1, 0) $, then
$ F(|u|,u) = F(-u,u) = H(u) \coloneqq g(u) + \frac{u}{2} \ln(1 - u^2) - \ln(1 + u^3) $.
By calculating the derivatives of this function, similarly as in
\eqref{eq:50}--\eqref{eq:60}, we find that
$ H(0) = H'(0) = 0 $ and $ H''(u) = R''(|u|) > 0 $.
So, $ M(1,b,u) = F(b,u) > F(|u|,u) = F(-u,u) = H(u) > 0 $. 

Therefore, \eqref{eq:M} holds.
The proof is complete.
\end{proof}

Finally, we note that the difference in maximum rates becomes larger as
Bob's speed grows.
In fact, as $ \beta \to 1 $, we have $ \overline{C}_A \to \infty $ and
$ \overline{C}_B \to 0 $.

\section{Concluding remarks}

We close the paper with a brief discussion on the more general form of the
Jarett--Cover conjecture corresponding to the case when Bob's speed is given
by an arbitrary function $ \beta(\tau) \colon [0, T_B] \to [0, 1) $.
This function may, for physical reasons, be assumed to be continuous, although
we believe the conjecture is true for any function nice enough for all the
relevant integrals to exist.

Unfortunately, the proof presented in the previous section does not translate
to this setting in a straightforward way.
In an attempt to resolve it, it may be instructive to focus first on special
cases, such as the purely radial motion (when $ |\beta_r(\tau)| = \beta(\tau) $),
or the circular motion around Alice (when $ \beta_r(\tau) = 0 $).
Another special case worth investigating is the one where the transmitter is
unrestricted in terms of the frequency bandwidth it may use ($ W = \infty $),
in which case the expressions for the maximum transmission rates take a simpler
form:%
\begin{subequations}
\label{eq:capinf}
\begin{align}
\label{eq:capinfA}
	 \lim_{W \to \infty} \overline{C}_A
	  &= \frac{\log e}{\eta} \frac{1}{T_A} \int_{0}^{T_A} \alpha_A(t) P_A(t) \,dt ,  \\
\label{eq:capinfB}
	 \lim_{W \to \infty} \overline{C}_B
	  &= \frac{\log e}{\eta} \frac{1}{T_B} \int_{0}^{T_B} \alpha_B(\tau) P_B(\tau) \,d\tau .
\end{align}
\end{subequations}
As before, $ P_A(t), P_B(\tau) $ are nonnegative functions with average
value $ \overline{P} $ that maximize \eqref{eq:capinfA} and \eqref{eq:capinfB},
respectively, and $ \alpha_A(t), \alpha_B(\tau) $ are given by \eqref{eq:doppler}.
Again, the question we are asking is whether the quantity in \eqref{eq:capinfA}
is larger that that in \eqref{eq:capinfB}, for any $ \overline{P} $ and $ \eta $.
Note that, since the capacities $ \overline{C}_A $ and $ \overline{C}_B $
are continuous functions of the bandwidth $ W $, this would automatically
imply that $ \overline{C}_A > \overline{C}_B $ for any ``large enough'' $ W $.

\vspace{5mm}
\textbf{Acknowledgment:}
This research was funded by the European Research Council (ERC) through the
European Union's Horizon 2020 Research and Innovation programme (Grant agreement no.~101003431).
M. Kova\v{c}evi\'{c} was also supported by the Science Fund of the Republic of
Serbia through the Serbian Science and Diaspora Collaboration Program (project
no.~6473027), and by the Secretariat for Higher Education and Scientific Research
of the Autonomous Province of Vojvodina (project no.~142-451-2686/2021).


\end{document}